\documentclass{article}
\usepackage[nonatbib, preprint]{neurips_2024}

\usepackage[utf8]{inputenc} %
\usepackage[T1]{fontenc}    %
\usepackage{hyperref}       %
\usepackage{url}            %
\usepackage{booktabs}       %
\usepackage{amsfonts}       %
\usepackage{nicefrac}       %
\usepackage{microtype}      %
\usepackage{xcolor}         %
\usepackage{graphicx}
\usepackage{amsmath}
\usepackage{multicol, multirow}
\usepackage{theorem}
\usepackage{wrapfig}
\usepackage{subcaption}
\usepackage[ruled,vlined]{algorithm2e}
\usepackage{listings}
\newtheorem{assumption}{Assumption}
\newtheorem{theorem}{Theorem}
\newtheorem{lemma}{Lemma}

\newtheorem{proof}{Proof}
\title{Training-Free Multi-Step Audio Source Separation}

\author{%
  Yongyi Zang\\
  Independent Researcher\\
  Seattle, WA, USA\\
  \texttt{zyy0116@gmail.com} \\
  \And
  Jingyi Li \\
  University of Illinois Urbana-Champaign \\
  Champaign, IL, USA \\
  \texttt{jingyi49@illinois.edu} \\
    \And
  Qiuqiang Kong\thanks{Corresponding author.} \\
  The Chinese University of Hong Kong \\
  Hong Kong ASR, China \\
  \texttt{qqkong@ee.cuhk.edu.hk} \\
}

\begin{document}

\maketitle

\begin{abstract}
Audio source separation aims to separate a mixture into target sources. Previous audio source separation systems usually conduct one-step inference, which does not fully explore the separation ability of models. 
In this work, we reveal that pretrained one-step audio source separation models can be leveraged for multi-step separation without additional training. We propose a simple yet effective inference method that iteratively applies separation by optimally blending the input mixture with the previous step's separation result. At each step, we determine the optimal blending ratio by maximizing a metric. We prove that our method always yield improvement over one-step inference, provide error bounds based on model smoothness and metric robustness, and provide theoretical analysis connecting our method to denoising along linear interpolation paths between noise and clean distributions—a property we link to denoising diffusion bridge models. Our approach effectively delivers improved separation performance as a "free lunch" from existing models.
Our empirical results demonstrate that our multi-step separation approach consistently outperforms one-step inference across both speech enhancement and music source separation tasks, and can achieve scaling performance similar to training a larger model, using more data, or in some cases employing a multi-step training objective. These improvements appear not only on the optimization metric during multi-step inference, but also extend to nearly all non-optimized metrics (with one exception). We also discuss limitations of our approach and directions for future research.
\end{abstract}

\section{Introduction} \label{sec:intro}

\begin{quote}
\textit{``In solving a problem, we often introduce auxiliary problems that were not part of the original statement but whose solution facilitates the solution of the original.''} 

\hfill \textit{--- Herbert Simon~\cite{simon1996sciences}}
\end{quote}

The remarkable progress in deep learning has been largely driven by scaling model capacity and training data, yielding breakthroughs across natural language processing, computer vision, and speech recognition. While traditional improvements follow \emph{scaling laws} \cite{kaplan2020scaling} focused on training-time resources, recent work has explored a complementary paradigm: \emph{inference-time scaling} \cite{ma2025inference}. This approach allocates additional compute during inference to enhance model performance without requiring larger architectures or expanded training datasets.

Inference-time scaling has proven particularly effective in domains with inherent multi-step structures. Diffusion models \cite{ho2020denoising, sohl2015deep} and flow matching \cite{lipman2022flow, liu2022flow} exemplify this paradigm, achieving superior performance through iterative refinement. However, these models require specialized training procedures and substantial computational overhead. This limitation motivates the development of \emph{training-free} methods that enable models trained with one-step objectives to benefit from multi-step inference. Recent successes include chain-of-thought prompting \cite{wei2022chain} in language models and process reward-based reasoning \cite{lightman2023let, luo2024improve, gao2023scaling} for selecting optimal inference paths.

In this work, we introduce the first training-free inference-time scaling method for audio source separation—the challenging task of isolating target audio signals from complex mixtures containing environmental noise, musical interference, or overlapping speech. Audio source separation is fundamental to numerous applications including hearing aids, music production, and speech recognition systems. While existing approaches train models to directly map noisy inputs to clean outputs in a single step, we demonstrate that these models possess latent capabilities for iterative refinement that can be unlocked without additional training.

Our key insight stems from the observation that standard training procedures for audio separation models—which synthesize noisy mixture through clean signals mixed with random interference signals at varying levels—inadvertently create models capable of operating across diverse noise conditions. We exploit this property through a novel inference strategy: at each step, we generate multiple candidate solutions by remixing the current output with the original noisy input at different ratios. These candidates are processed through the model and ranked using task-specific quality metrics, with the best output selected for the next iteration. Through this simple approach, we effectively transform one-step models into multi-step systems, achieving performance gains without architectural modifications or retraining.

We make the following contributions:
\begin{itemize}
    \item We propose the first training-free inference-time scaling method for audio source separation, demonstrating consistent improvements across diverse separation tasks.
    \item We provide theoretical analysis establishing performance bounds of our method based on network smoothness characterized by their Lipschitz constants and metric robustness (Section~\ref{sec:method}), and connecting our surprisingly simple approach to the training process of diffusion denoising bridge models (Section~\ref{sec:bridge_connection}), providing justification for its effectiveness.
    \item We empirically validate our method across two audio separation subdomains: speech enhancement (Section~\ref{ssec:se}) and music source separation (Section~\ref{ssec:mss}).
    \item We open-source our implementation under the Apache 2.0 license at \url{https://github.com/yongyizang/TrainingFreeMultiStepASR}, enabling reproducibility and further research.
\end{itemize}

Our results demonstrate that inference-time scaling represents a powerful yet underexplored paradigm for audio processing, offering immediate performance gains for existing models while providing insights for future architectural designs.

\section{Related works}

\textbf{Audio Source Separation.} Modern audio source separation has been dominated by deep learning approaches operating in either time-domain or time-frequency domains. U-Net architectures remain influential, from Wave-U-Net's pioneering direct waveform processing \cite{stoller2018wave} to recent variants like ResUNet \cite{kong2021decoupling} achieving 8.98 dB signal-to-distortion ratio (SDR) on vocals, and DTTNet~\cite{chen2024music} reaching 10.12 dB cSDR with 86.7\% fewer parameters. Transformer architectures have emerged for modeling long-range dependencies, with SepFormer \cite{subakan2021attention} setting state-of-art performance on multiple audio separation benchmarks cite{hershey2016deep}, and BS-RoFormer \cite{lu2024music} winning SDX23 \cite{fabbro2023sound} with 9.80 dB average SDR. Band-specific processing strategies, exemplified by Band-Split RNN \cite{yu2022high}, have shown significant gains by explicitly modeling frequency-dependent characteristics. While these one-step models offer computational efficiency, they typically trade separation quality for speed, and show limitations as scaling model parameters become ineffective due to lack of real-world data \cite{keren2018scaling, zhang2024beyond}.

\textbf{Multi-Step Audio Source Separation.} Multi-step approaches consistently demonstrate superior separation quality, albeit at higher computational costs. Diffusion models have established new performance benchmarks, with Multi-Source Diffusion Models \cite{mariani2023multi}, DiffSep \cite{scheibler2023diffusion}, DOSE \cite{tai2023dose} and SGMSE \cite{richter2023speech} achieving state-of-the-art results but requiring 50-200 inference steps. Flow matching techniques like FlowSep \cite{yuan2025flowsep} offer efficiency gains, achieving comparable quality with only 10-20 steps. While multi-step methods typically achieve better performance, they require 10-100 times more computation during inference. Recent work has focused on reducing step counts through techniques like progressive distillation and adversarial conditional diffusion distillation \cite{kaneko2024fastvoicegrad}, yet these approaches still require specialized multi-step training procedures.

\textbf{Training-Free Multi-Step Inference.} While multi-step models require specific training procedures, recent advances in training-free inference offer promising alternatives. In natural language processing, Chain-of-Thought \cite{wei2022chain} and Self-Consistency \cite{wang2022self} approaches have shown that models can benefit from structured multi-step reasoning without specialized training. Audio-specific training-free approaches include zero-shot separation through query-based learning \cite{chen2022zero} and unsupervised separation by repurposing pretrained models \cite{manilow2022improving}. Recent work on test-time scaling \cite{muennighoff2025s1, snell2024scaling} demonstrates that performance can improve with increased inference-time computation, though these concepts have not been systematically explored for audio source separation. Process Reward Models \cite{zhang2025lessons} offer mechanisms for evaluating intermediate steps, but their application to audio remains unexplored. Our work bridges this gap by introducing a training-free multi-step inference method specifically designed for audio source separation, leveraging the inherent properties of models trained with data augmentation to achieve multi-step improvements without retraining.

\section{Method}
\label{sec:method}
We propose a simple refinement mechanism shown in Algorithm~\ref{alg:iter-refine} that improves separation quality through iterative processing. We denote a pretrained separation model as $f(\cdot)$. We denote the noisy mixture as $x_0 \in \mathbb{R}^{L}$, where $ L $ is the number of audio samples. we start from an initial separation $y_0 = f(x_0)$ and perform the separation for $T$ refinement steps. We propose to improve the separation quality of $ y_{0} $ by using a multi-step update strategy as follows.

At each step $t$ ranging from 1 to $T$, we create a new input mixture $x_{t}$ by:
\begin{equation}\label{eq:blend}
   x_t = r_{t} x_0 + (1 - r_{t}) y_{t-1},
\end{equation}
\noindent where $x_{0}$ is the original input mixture, $ y_{t-1} $ is the separation result at the $(t-1)$-th step, and $r_{t} \in [0, 1] $ is a variable to be optimzed. Equation (\ref{eq:blend}) shows that the new input mixture $x_{t}$ is the blending between the input mixture $ x_{0} $ and the separated source $ y_{t-1} $. 

Directly optimizing $r_{t}$ is a challenging problem. To address this problem, We linearly sample evenly spaced $K$ values between 0 and 1 and denote them as $r_{t}^{(k)}$. We denote $ x_{t}^{(k)} $ corresponds to $r_{t}^{(k)}$ calculated by (\ref{eq:blend}). Then, we search for the optimal $r_{t}^{*}$ that optimize a metric:
\begin{equation}\label{eq:reward}
   r_{t}^{*} = \arg\max_{k \in \{1,\ldots,K\}} R(f(x_t^{(k)})).
\end{equation}
\noindent We define $ R $ as a metric to evaluate the separation quality, such as SI-SNR and DNSMOS \cite{reddy2021dnsmos}. A larger metric indicates better separation quality. By this means, we search the optimal $ r^{*} $ to constitute $ x_{t}^{*} $. The optimal separation result is calculated by $ y_{t}=f(x_{t}^{*}) $.

\begin{algorithm}[t!]
\caption{Multi-step Inference for One-step Audio Separation Models}
\label{alg:iter-refine}
\KwIn{Mixture signal $x_0$, separation model $f$, total steps $T$, number of ratios $K$}
\KwOut{Final separated signal $y_T$}
$y_0 \leftarrow f(x_0)$\;
\For{$t\leftarrow 1$ \KwTo $T$}{
    Sample $K$ ratios evenly in $[0,1]$ and denote them by $r^{(k)}$\;
    Constitute $K$ inputs $x_{t}^{(k)} = r_{t}^{(k)} x_0 + (1 - r_{t}^{(k)}) y_{t-1}$\;
    Search the optimal $r_{t}^{*}$ by (\ref{eq:reward}) and constitute $x_{t}^{*}$ by (\ref{eq:blend})\;
    Calculate $y_{t} = f(x_{t}^{*})$\;
}
\Return $y_T$
\end{algorithm}

Then, we repeat updating $r_{t}^{*}$, $x_{t}^{*}$, and $ y_{t} $ for $ T $ steps to calculate the final step $ y_{T} $ as the separation result. 
As proven in Section 3.1, the separation quality $R(y_t)$ is non-decreasing with respect to the initial quality $R(y_0)$ when the optimal $r_t^*$ is chosen. As we prove in Section 3.2, if the estimated ratio $ r_{t} $ deviates from the optimal ratio with a small value, the upper bound deviation of metric $ R(y_{t}) $ is predictable.

\subsection{Non-Decreasing Property of the Separation Metric}

We first show that the multi-step separation increases the lower bound of the one-step separation as follows. 

\begin{theorem}[Lower Bound of Metrics]
\label{thm:error_bound}
There always exists an optimal $ x_{t}^{*} $ such that $ R(y_{t}) \geq R(y_{0}) $.
\end{theorem}

\begin{proof}
for all $t \in \{1,\ldots,T\}$,
\begin{align}
R\left(y_t\right) & =R\left[f\left(x_t^*\right)\right] \\
& = R\left[ f(r_{t}^{*}x_{0} + (1-r_{t}^{*})y_{t-1}) \right] \\
& \geq R\left[f\left(1 \cdot x_0 + 0\cdot y_{t-1}\right)\right] =R\left[f(x_0)\right]=R\left[y_0\right].
\end{align}
\end{proof}

Proof 1 shows that $ R(y_{t}) \geq R(y_{0}) $, indicating the multi-step separation increases the lower bound of one-step separation.

\subsection{Error Bound Analysis}

We now analyze the error when the estimated ratio $ r_{t} $ deviates from the optimal ratio. We will show $ \textit{Var}(R(y_{t})) $ is lower than a bound. We first introduce Lipschitz bound on derivative. 
\begin{lemma}[Lipschitz Bound on Derivative]
\label{lem:lipschitz_derivative}
If $f$ satisfies the Lipschitz condition $|f(x_1) - f(x_2)| \leq L|x_1 - x_2|$, then:
\begin{equation}
   |f'(x)| \leq L \quad \text{for all } x \text{ where } f \text{ is differentiable}.
\end{equation}
\end{lemma}
Then, we assume the separation model $ f(\cdot) $ and the evaluation metric $ R(\cdot) $ satisfy the following assumption.
\begin{assumption}
\label{ass:lipschitz}
The separation model $f(\cdot)$ is locally differentiable and Lipschitz continuous with constant $L_f$. The metric $R(\cdot)$ is locally differentiable and Lipschitz continuous with constant $L_r$. As in practice the separation model are often neural networks and common metrics are either continous ratios (SNR, SI-SDR, etc.) or neural networks (UTMOS, DNSMOS, etc.), we assume both are likely true in reality.
\end{assumption}

At each step $t$, following our algorithm notation, we have:
\begin{equation}
x_t = r_t x_0 + (1 - r_t) y_{t-1}.
\end{equation}

Due to the imperfect metric, the selected ratio $r_t$ may differ from the optimal ratio $r_t^*$. We model this uncertainty by assuming a gaussian distribution centered around $r_t^*$:
\begin{equation}
   r_t \sim \mathcal{N}(r_t^*, \varepsilon_r^2).
\end{equation}

Using the properties of normal distributions, we obtain:
\begin{equation}\label{eq:gaussian_x}
   x_t \sim \mathcal{N}(r_t^* x_0 + (1-r_t^*) y_{t-1}, (x_0 - y_{t-1})^2 \varepsilon_r^2),
\end{equation}
\noindent where $ \textit{Var}(x_{t}) = (x_0 - y_{t-1})^2 \varepsilon_r^2 $ has a dimension of number of audio samples.

\begin{theorem}[Error Bound]
\label{thm:error_bound}
Under Assumption~\ref{ass:lipschitz}, the variance of the metric function at step $t$ is bounded by:
\begin{equation}
   \text{Var}[R(y_t)] \leq L_f^2 L_r^2 (x_0 - y_{t-1})^2 \varepsilon_r^2
\end{equation}
\end{theorem}

\begin{proof}
Since $y_t = f(x_t)$, applying Lemma~\ref{lem:lipschitz_derivative} to both $f$ and $R$:
\begin{align}
   \text{Var}[R(y_t)] &\leq L_r^2 mean(\text{Var}(y_t))  \ \ \ \ \textcolor{gray}{\text{\# Lipschitz on $ R(\cdot) $}} \\
   &\leq L_r^2 \cdot L_f^2 \text{Var}(x_t) \ \ \ \ \textcolor{gray}{\text{\# Lipschitz on $ f(\cdot) $}} \\
   &= L_f^2 L_r^2 (x_0 - y_{t-1})^2 \varepsilon_r^2 \ \ \ \ \textcolor{gray}{\text{\# \ Eq. (\ref{eq:gaussian_x})}}
\end{align}
\end{proof}

Thus we have:
\begin{itemize}
    \item The error is proportional to the distance $(x_0 - y_{t-1})^2$ between the original mixture and current estimate, which typically decreases as iterations progress
    \item The bound depends quadratically on the Lipschitz constants $L_f$ and $L_r$ - i.e. the smoothness of both the separation network and metric.
    \item The bound also depends quadratically on the uncertainty of metric $\varepsilon_r^2$.
\end{itemize}

For well-trained models with small Lipschitz constants—achieved through techniques such as spectral normalization \cite{miyato2018spectralnormalizationgenerativeadversarial}, gradient regularization \cite{finlay2019lipschitzregularizeddeepneural}, or Lipschitz regularization \cite{gouk2020regularisationneuralnetworksenforcing}—the error bounds remain tight even with noisy metrics. 

While our error bound shows that the algorithm's behavior is controlled even with imperfect metrics, we can further analyze its convergence properties. Consider the expected improvement at each step. Even when the metric occasionally selects suboptimal candidates, as long as $\mathbb{E}[r_t]$ is sufficiently close to $r_t^*$ and the metric has positive correlation with true quality, the algorithm exhibits statistical convergence. Specifically, the expected change in true quality satisfies $\mathbb{E}[\Delta Q_t] > 0$ when averaged over the metric's noise distribution, where $Q_t$ represents the true (unobserved) quality at step $t$. Moreover, as iterations progress and $\|x_0 - y_{t-1}\|$ decreases, our error bound naturally tightens: $\text{Var}[R(y_t)] \propto (x_0 - y_{t-1})^2$. This self-stabilizing property means the algorithm becomes increasingly robust to metric noise as it approaches a solution.

\section{Results}
\label{sec:results}
We evaluate our iterative inference method across two audio separation tasks: speech enhancement \cite{benesty2006speech} (removing environmental noise) and music source separation \cite{cano2018musical} (isolating instruments or vocals). To simulate practical conditions, we use faster metric estimations or imperfect calculation methods during search when available, while reporting all standard metrics for evaluation. All experiments use $K = 10$ mixing ratio candidates per step and run for $T = 20$ total steps. We report metrics at steps 0 (baseline one-step inference), 1, 5, 10, and 20. We show results to four decimal places to illustrate progressive improvements, reverting to two decimal places when comparing with prior work following standard reporting practices in the field. We show more analysis, including evaluation results at all steps in Appendix. All experiments are conducted on a single NVIDIA RTX 4090.

\subsection{Speech Enhancement}
\label{ssec:se}
Following standard speech enhancement evaluation protocols \cite{zhang2024urgent, reddy2021icassp}, we evaluate our method's performance under both non-blind (the clean target is known) and blind (the clean target is not known) scenarios.

\textbf{Non-blind Evaluation.} In non-blind scenarios, we evaluate on synthetically constructed noisy mixtures where clean reference signals are available. This enables the use of intrusive metrics that directly compare enhanced outputs against ground-truth clean audio, providing precise quantitative assessment. We utilize the test set of VCTK-DEMAND dataset, which combines clean speech from VCTK \cite{veaux2013voice} with diverse noise samples from DEMAND \cite{thiemann2013diverse} to form 824 test samples. While synthetic datasets enable controlled evaluation conditions and perfect reference signals, they may not fully capture the complexity of real-world acoustic environments.

\textbf{Blind Evaluation.} To address the domain gap between synthetic and real-world conditions, we additionally evaluate on naturally noisy recordings where clean references are unavailable. This necessitates non-intrusive metrics that assess perceptual quality without reference signals. We employ the DNS Challenge v3 blind test set [3], which contains 600 real-world recordings with authentic noise characteristics. This dual evaluation strategy ensures both rigorous quantitative assessment and practical performance validation.

\textbf{Search and Evaluation Metrics.} We evaluate our approach using both intrusive and blind metrics. For intrusive evaluation, we report Perceptual Evaluation of Speech Quality (PESQ), Short-Time Objective Intelligibility (STOI), and Scale-Invariant Signal-to-Noise Ratio (SI-SNR). During inference, we conduct search using a faster PESQ estimator~\footnote{\url{https://github.com/audiolabs/torch-pesq}} instead of exact PESQ calculations \cite{kim2019end} to simulate realistic conditions with imperfect search. For blind evaluation, we employ DNSMOS-P.808, which provides BAK (background), SIG (signal), OVL (overall), and MOS (mean opinion score) predictions, alongside UTMOS \cite{saeki2022utmos}. We use UTMOS for search during inference.

All experiments utilize the pretrained BSRNN~\footnote{Available at \url{https://github.com/Emrys365/se-scaling}} model \cite{luo2023music} from \cite{zhang2024beyond}. These models were trained as one-step speech enhancement systems with varying sizes and computational complexities on a combined dataset from VCTK-DEMAND, DNS Challenge V3 training set, and WHAMR!~\cite{maciejewski2020whamr}. We use the ``medium'' variant as our base model for multi-step inference, and compare against the ``large'' and ``xlarge'' variants to evaluate our method against simply scaling model size. For context, the ``large'' variant contains 3.73 times more parameters and requires 3.96 times more multiply–accumulate (MAC) operations compared to the ``medium'' variant, while the ``xlarge'' variant has 6.16 times more parameters and requires 7.85 times more MAC operations. We also include a state-of-art diffusion-based enhancement model SGMSE+~\cite{richter2023speech} with 30 steps during inference, and use metrics reported in its paper for comparison.

Table~\ref{tab:se_results} presents our evaluation results. On intrusive metrics, our method achieves substantial PESQ improvement from 3.21 to 3.28 with just the first step of our proposed inference method, despite using an approximated PESQ metric during search. Further improvements to 3.29 are observed at inference step 20. For STOI, we observe improvements but not monotonic increases, with optimal performance at both steps 5 and 20. In contrast, SI-SNR shows a consistent decline with additional steps. We attribute this divergent behavior to the limited correlation between SI-SNR and perceptual metrics like PESQ/STOI, as the latter focus on subjective quality while SI-SNR measures objective signal fidelity. 

To contextualize our results, we compare our approach against simply using larger models. Interestingly, our ``medium'' model at inference step 0 outperforms the ``large'' variant at step 0, aligning with observations in \cite{zhang2024beyond} that highlight limitations of model scaling for speech enhancement tasks. While the ``xlarge'' variant performs worse on VCTK-DEMAND, it achieves better metrics on DNS Challenge V3 and superior UTMOS scores compared to our 20-step approach. Nevertheless, our method demonstrates better performance on DNSMOS signal (SIG), overall (OVRL), and MOS predictions, with just one inference step employed with our method sufficient to improve SIG and OVRL scores despite optimizing solely for UTMOS during search.

When comparing against a diffusion-based model trained on multi-step objective~\cite{richter2023speech}, we observe that on VCTK-DEMAND, its intrusive metrics fall behind even our medium model with one-step inference, consistent with that paper's own findings. On DNS Challenge V3, while the diffusion model shows better performance in OVRL and SIG, our approach achieves comparable MOS performance with just the first step of our proposed method\footnote{As previously noted, this comparison is made with only 2 decimal points following previous works.}. For the BAK metric, our medium variant already outperforms the diffusion-based model. This highlights limitations for models trained with multi-step objectives, and further validates the effectiveness of our proposed method.

\begin{table}[h]
  \footnotesize
  \centering
  \begin{subtable}[t]{0.4\textwidth}
    \centering
    \begin{tabular}{cccc}
    \toprule
        \# Step & PESQ            & STOI            & SI-SNR           \\
        \midrule
        0       & 3.2042          & 0.9577          & \textbf{19.2717} \\
        1       & 3.2764          & 0.9588          & 18.7694          \\
        5       & 3.2832          & \textbf{0.9590} & 18.7159          \\
        10      & 3.2864          & 0.9589          & 18.6976          \\
        20      & \textbf{3.2867} & \textbf{0.9590} & 18.6880          \\
        \midrule
        L   & 3.0951          & 0.9575          & 18.7940          \\
        XL & 3.0863          & 0.9574          & 18.5120          \\
        \cite{richter2023speech} & 2.93          & -          & 17.3          \\
        \bottomrule
        \end{tabular}
    \caption{Intrusive metrics on VCTK-DEMAND}
  \end{subtable}
  \hfill                               %
  \begin{subtable}[t]{0.58\textwidth}
    \centering
    \begin{tabular}{cccccc}
    \toprule
    \# Step & UTMOS           & BAK             & SIG             & OVRL            & MOS             \\
    \midrule
    0       & 2.3101          & 3.9444          & 3.2476          & 2.9591          & 3.6158          \\
    1       & 2.3355          & 3.9590          & 3.2776          & 2.9894          & 3.6353          \\
    5       & 2.3379          & 3.9580          & 3.2820          & 2.9924          & \textbf{3.6378} \\
    10      & 2.3380          & 3.9581          & 3.2821          & 2.9925          & \textbf{3.6378} \\
    20      & 2.3380          & 3.9586          & \textbf{3.2828} & \textbf{2.9931} & \textbf{3.6378} \\
    \midrule
    L   & 2.2974          & 3.8986          & 3.2349          & 2.9277          & 3.5999          \\
    XL & \textbf{2.3487} & \textbf{3.9616} & 3.2757          & 2.9893          & 3.6359     \\
    \cite{richter2023speech} & - & 3.82 & 3.42          & 3.04          & 3.64     \\
    \bottomrule
    \end{tabular}

    \caption{Non-intrusive metrics on DNS Challenge V3}
  \end{subtable}

  \caption{Performance comparison across inference‑time steps for speech enhancement. Best metrics are denoted in \textbf{bold}; ``L'' and ``XL'' refers to the one-step inference result for the ``large'' and ``xlarge'' variants of the model. Unavailable metrics due to lack of report in~\cite{richter2023speech} is denoted with a dash (-).}
  \label{tab:se_results}
\end{table}
    
\subsection{Music Source Separation}
\label{ssec:mss}
\begin{table}[h]
\footnotesize
\centering
\begin{tabular}{ccccccccc}
\toprule
\multirow{2}{*}{\# Step} & \multicolumn{4}{c}{uSDR (dB)}                                          & \multicolumn{4}{c}{cSDR (dB)}                                         \\
                         & Vocals           & Bass            & Drums           & Other           & Vocals          & Bass            & Drums           & Other           \\
                         \midrule
0                        & 10.2514          & 7.0904          & 7.6070          & 6.1259          & 8.7523          & 5.0513          & 7.0081          & 7.1150          \\
1                        & 10.4129          & 7.3754          & 7.9163          & 6.4366          & 9.1254          & 5.5758          & 7.3251          & 7.4163          \\
5                        & 10.4428          & 7.4699          & 8.0079          & \textbf{6.4615} & 9.0821          & 5.7376          & \textbf{7.3574} & \textbf{7.4279} \\
10                       & 10.4493          & 7.5083          & 8.0160          & 6.4500          & 9.0962          & 5.7180          & 7.3399          & 7.3194          \\
20                       & \textbf{10.4542} & \textbf{7.5353} & \textbf{8.0401} & 6.4502          & \textbf{9.1123} & \textbf{5.7532} & 7.3045          & 7.3123   \\
\bottomrule \\
\end{tabular}
  \caption{Performance comparison across inference‑time steps for music source separation. Best metrics are denoted in \textbf{bold}.}
  \label{tab:mss_results}
\end{table}

For music source separation, we adopt the MUSDB18-HQ \cite{musdb18-hq} test set, a standard benchmark comprising 50 songs spanning various musical genres \cite{fabbro2023sound} and provides mixture and four evaluation stems: vocals, bass, drums and other. 

\textbf{Search and Evaluation Metrics.} For music source separation evaluation, we employ two standard signal-to-distortion ratio (SDR) metrics on stereo signals: Utterance-level SDR (uSDR), which computes the mean SDR across all test songs; and chunk-level SDR (cSDR), which calculates the median of per-song values, where each song's value is the median SDR of all 1-second segments. During inference, we use a modified SDR calculation that differs from the evaluation metrics in two key aspects: (1) we use 6-second chunks with 50\% overlap instead of 1-second segments, and (2) we compute SDR independently for each channel rather than jointly on stereo signals. This search metric provides similar trends to the final evaluation metrics yet is not exactly the same. All final results are reported using the standard uSDR and cSDR metrics on stereo signals.

We employ the released checkpoints from \cite{chen2024music} and present our evaluation results in Table~\ref{tab:mss_results}. We observe slight metric discrepancies compared to those reported in \cite{chen2024music}; this phenomenon is common in music source separation due to variations in evaluation packages \cite{ward2018bss}. Since our analysis focuses on relative differences across time steps within our own evaluation framework, we consider our results internally consistent and reliable. Our results exhibit a similar trend to speech enhancement: the first step of our proposed inference method yields the most substantial improvement, with vocals uSDR showing a gain of 0.16 dB. To contextualize this improvement, the original paper \cite{chen2024music} reported a 0.11 dB gain when modifying network parameters to achieve 10× inference time, and a 0.07 dB gain from additional data annotation for fine-tuning. All inference steps conducted using our proposed method show consistent improvement over one-step inference.

We find that uSDR and cSDR do not perfectly correlate, and performance does not increase monotonically across steps for all stems. While the ``Vocals'' and ``Bass'' stem shows consistent improvement, ``Drums'' and ``Other'' achieve optimal performance at intermediate step 5. Notably, our method yields improvements for ``Vocals,'' ``Bass,'' and ``Drums'' stems that are comparable to or exceed those achieved by the iterative fine-tuning method in the BSRNN paper on an additional 1,750 songs, 17.5 times larger than the training dataset \cite{luo2023music}. Importantly, all multi-step inference results using our proposed method outperform one-step inference.

\section{Connection to Bridge Models}
\label{sec:bridge_connection}

Why is our simple method so effective in practice? In this section, we connect our proposed method to properties of denoising diffusion bridge models (DDBMs) \cite{zhou2023denoising}. We show that audio source separation models, through their often-employed augmentation-based training procedures, naturally exhibit DDBM-like properties that our inference method exploits.

Audio source separation faces a fundamental challenge: real-world mixtures rarely come with their clean signal counterparts. To address this, the field has universally adopted \emph{data mixing} during training. As examples, the DNS Challenge \cite{reddy2020interspeech, reddy2021icassp} dynamically mixes noise and speech at various SNRs \cite{reddy2019scalable}, and music separation model practitioners mixes sources from different songs \cite{jeon2024does, defossez2019music}.

The data mixing process is straightforward: sample a clean signal $p\sim\mathcal{P}$, a noise signal $q\sim\mathcal{Q}$, and a mixing coefficient $\sigma \in [0, 1]$~\footnote{The mixing coefficient is often in the form of SNR, which can be expressed as a gain level between 0 and 1.}. The mixture is then:

\begin{equation}
y = \sigma p + (1 - \sigma) q
\end{equation}

And the audio source separation model learns to recover $p$ by minimizing:
\begin{equation}
\label{eq:lsep}
\mathcal{L}_{\text{sep}}(\cdot) = \mathbb{E}_{p,q,\sigma} \left[ |f(y) - p|^2 \right]
\end{equation}

We now show that this standard mixing process enables a special case of bridge diffusion. 

\begin{proof} Audio mixtures are points along a linear bridge. Consider a diffusion bridge that connects noise samples from $\mathcal{Q}$ to clean samples from $\mathcal{P}$. For a linear bridge, the interpolation at time $t \in [0, T]$ is:

\begin{equation}
x_t = \frac{T-t}{T} x_0 + \frac{t}{T} x_T
\end{equation}

where $x_0 \sim \mathcal{Q}$ and $x_T \sim \mathcal{P}$. Consider $T = 1$, $t = \sigma$, $x_0 = q$, and $x_1 = p$, we get our mixture:
\begin{equation}
x_\sigma = (1-\sigma) q + \sigma p = y
\end{equation}
\end{proof}

In the DDBM framework, models learn by score matching:
\begin{equation}
\mathcal{L}_{\text{DDBM}}(\cdot) = \mathbb{E}{t, x_0, x_T} \left[ w(t) |s_\omega(x_t, x_T, t) - \nabla_{x_t} \log q(x_t | x_0, x_T)|^2 \right]
\end{equation}

where $w(t)$ is an arbitary time-dependent weighting function that controls the importance of different time steps in the loss. For our linear bridge, the conditional distribution $q(x_t | x_0, x_T)$ is a Dirac delta at the interpolation point. As the Dirac delta doesn't have a well-defined gradient, we consider a smoothed version, a Gaussian centered at the interpolation point with small variance $\epsilon$:
\begin{equation}
q_\epsilon(x_t | x_0, x_T) = \mathcal{N}(x_t; (1-t)x_0 + t x_T, \epsilon^2 I)
\end{equation}
we thereby derived a score:
\begin{equation}
\nabla_{x_t} \log q_\epsilon(x_t | x_0, x_T) = -\frac{1}{\epsilon^2}(x_t - ((1-t)x_0 + t x_T))
\end{equation}

In audio separation, it's natural to parameterize the score using the clean signal prediction and noise level. Consider:
\begin{equation}
s_\omega(y, p, \sigma) = \frac{1}{\epsilon^2\sigma}(y - \hat{p}(y))
\end{equation}
where $\hat{p}(y)$ is the model's estimate of the clean signal. The justification for this formulation is in Appendix. We now substitute this into the DDBM objective:
\begin{align}
\mathcal{L}_{\text{DDBM}}(\cdot) &= \mathbb{E}_{\sigma, p, q} \left[ w(\sigma) \left| \frac{1}{\epsilon^2\sigma}(y - \hat{p}(y)) - \nabla_y \log q\epsilon(y | q, p) \right|^2 \right] \\
&= \mathbb{E}_{\sigma, p, q} \left[ w(\sigma) \left| \frac{1}{\epsilon^2\sigma}(y - \hat{p}(y)) + \frac{1}{\epsilon^2}(y - ((1-\sigma)q + \sigma p)) \right|^2 \right]
\end{align}
As $y = (1-\sigma)q + \sigma p$ by definition, the second term simplifies to zero. Thus:
\begin{align}
\mathcal{L}_{\text{DDBM}}(\cdot) &= \mathbb{E}_{\sigma, p, q} \left[ w(\sigma) \left| \frac{1}{\epsilon^2\sigma}(y - \hat{p}(y))\right|^2 \right] \\
&= \mathbb{E}_{\sigma, p, q} \left[ \frac{w(\sigma)}{\epsilon^4\sigma^2} |\hat{p}(y) - y|^2 \right]
\end{align}
Consider the case $w(\sigma) = \sigma^2$:
\begin{align}
\mathcal{L}_{\text{DDBM}}(\cdot) &= \mathbb{E}_{\sigma, p, q} \left[ \frac{1}{\epsilon^4} |y - \hat{p}(y)|^2 \right]
\end{align}

Recall Equation~\ref{eq:lsep}, and as $\frac{1}{\epsilon^4}$ is a constant, we establish a equivalence between the DDBM training objective and the audio source separation network, i.e. $\mathcal{L}_{\text{DDBM}}(\cdot) \propto \mathcal{L}_{\text{sep}}(\cdot)$

This connection to bridge models explains our inference method's effectiveness. Audio separation models trained with data mixing implicitly learn denoising capabilities along the interpolation path between noise and clean distributions. The equivalence $\mathcal{L}_{\text{DDBM}}(\cdot) \propto \mathcal{L}_{\text{sep}}(\cdot)$ shows that separation models $f(\cdot)$ can denoise inputs at arbitrary points along this linear bridge.

Our candidates $x_t^{(k)} = r_k \cdot x_0 + (1 - r_k) \cdot y_{t-1}$ sample different points along this bridge, each corresponding to a different noise level. The model handles these intermediate noise levels naturally because they lie within the interpolation space encountered during training. 

Our method therefore leverages an inherent property of audio separation models: training with data mixing creates models capable of iterative refinement. By generating and selecting candidates along the learned denoising manifold, our inference algorithm utilizes this capability without requiring additional training. This insight suggests that models trained with interpolative augmentation may possess latent potential for training-free performance improvements through appropriate inference strategies. We present additional theoretical analysis in Appendix.

\section{Discussions, Limitations and Future Work}
\label{sec:discussions}
Our method demonstrates effectiveness in practice, while several limitations warrant discussion. First, our approach requires a metric during inference. For non-intrusive metrics (Section~\ref{ssec:se}), neural networks can be readily employed. However, in other scenarios, we utilized intrusive metrics that require ground-truth audio, which is often unavailable in real-world applications. Exploring neural estimators of these metrics~\cite{kumar2023torchaudio} remains promising future work. Secondly, We observed that all metrics improved upon one-step inference except SI-SNR on VCTK-DEMAND. This anomaly likely stems from the synthetic nature of this dataset, which has been criticized for its high SNR range~\cite{salman2024towards} and limited noise types~\cite{maciejewski2020whamr}. Despite these limitations, we included VCTK-DEMAND due to its prevalence in speech enhancement literature~\cite{das2021fundamentals}. Future work should evaluate our method on more realistic benchmarks across diverse audio separation scenarios, such as the improved dataset proposed by the original authors~\cite{li2021dds}. Finally, while our method employs multiple inference steps that could increase computational cost, many diffusion-based approaches already utilize multi-step inference. Nevertheless, developing optimized sampling schedules and accelerated inference methods represents an important direction for future research.

We acknowledge that enhanced source separation capabilities could allow easier sourcing of training data for generative models, potentially enabling audio deepfake creation~\cite{singfake}. We urge the community to exercise responsible judgment when using audio source separation techniques.

\section{Conclusions}

We presented a training-free inference-time scaling method that transforms pretrained one-step audio source separation models into multi-step systems. Our approach achieves performance gains by iteratively refining outputs through strategic remixing and quality-based selection—requiring only minor computational overhead compared to specialized multi-step models. Across speech enhancement and music separation tasks, we demonstrated consistent improvements over one-step inference in all but one case. Our theoretical analysis reveals two key insights. First, we establish error bounds showing that performance depends quadratically on network smoothness (Lipschitz constants) and metric reliability, with a self-stabilizing property that reduces sensitivity to metric noise as iterations progress. Second, we connect our method to denoising diffusion bridge models, demonstrating that standard data augmentation procedures inadvertently create models capable of iterative refinement along learned denoising manifolds. This bridge model perspective explains why our simple approach succeeds: separation models trained with interpolative mixing naturally possess multi-step capabilities that can be exploited without retraining.

\small

\bibliographystyle{plain}
\bibliography{main}

\clearpage

\appendix
\section{Additional Details for Speech Enhancement}
\subsection{Results on VCTK-DEMAND}
We report the per-step PESQ, STOI and SI-SNR results for VCTK-DEMAND in Figure~\ref{fig:vctk-pesq}, Figure~\ref{fig:vctk-stoi} and Figure~\ref{fig:vctk-sisnr}. We notice the results become stable after two inference steps; the first inference step yield the largest improvement.

\begin{figure}[h]
\centering
\includegraphics[width=\linewidth]{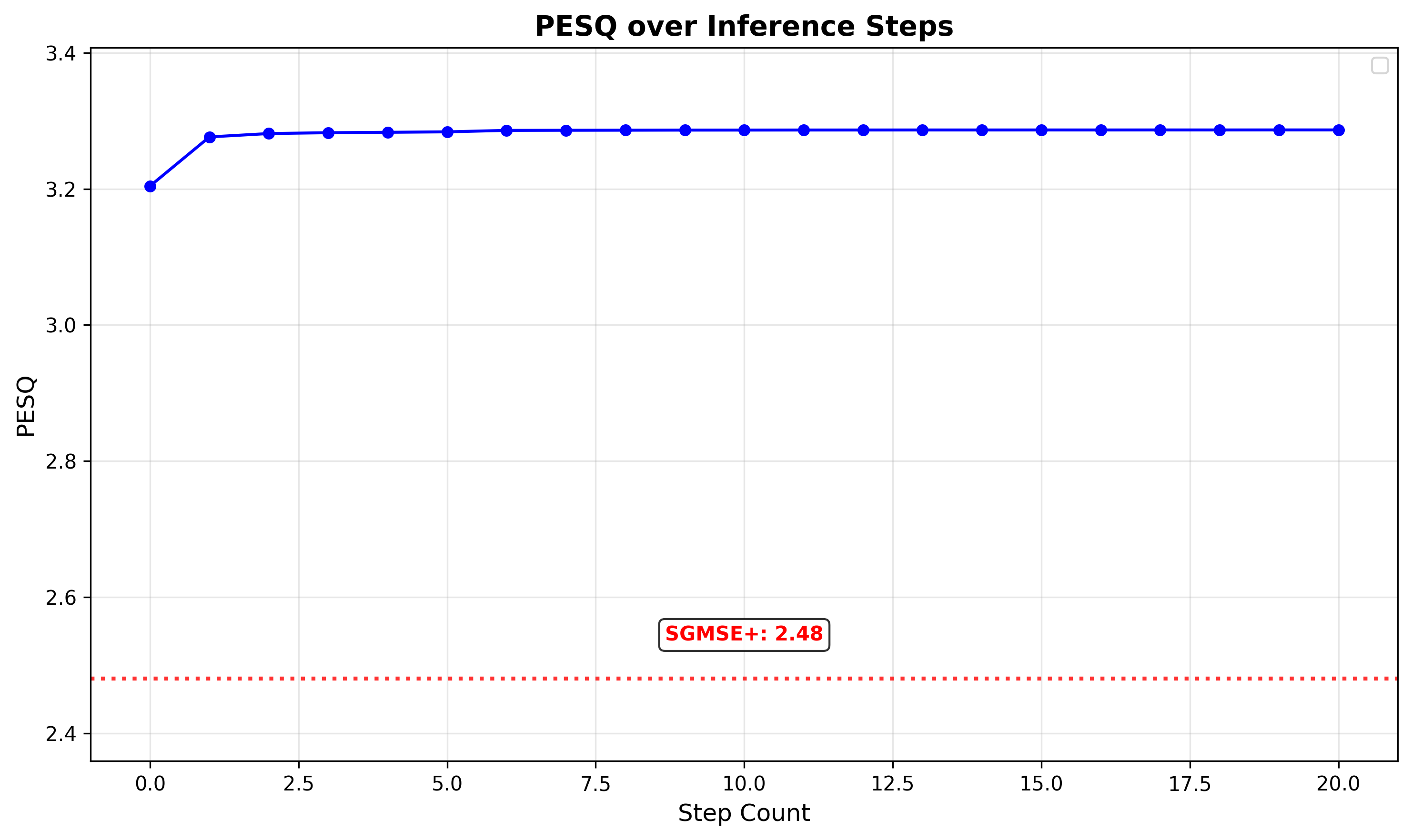}
\caption{PESQ across all inference steps for VCTK-DEMAND test set.}
\label{fig:vctk-pesq}
\end{figure}

\begin{figure}[h]
\centering
\includegraphics[width=\linewidth]{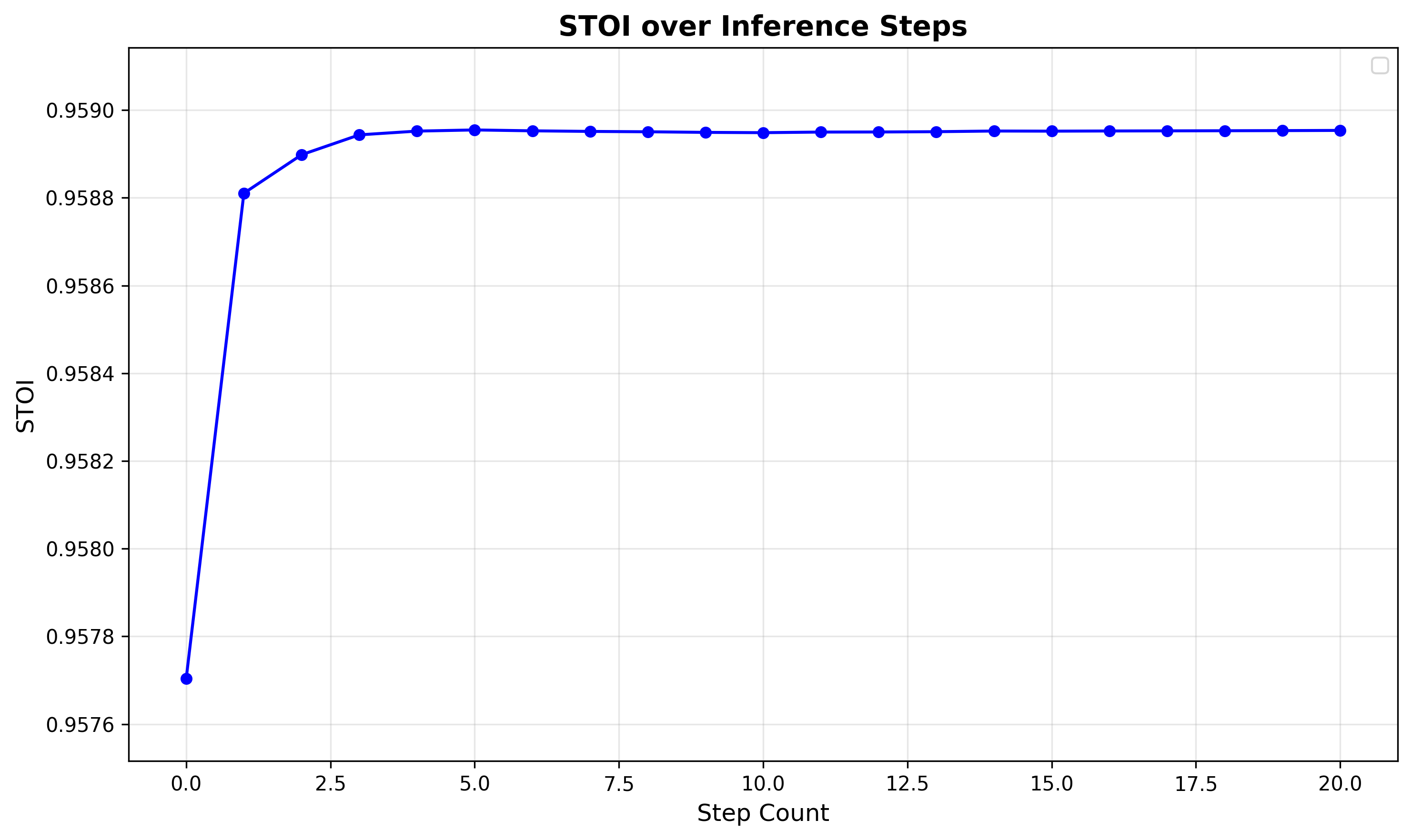}
\caption{STOI across all inference steps for VCTK-DEMAND test set.}
\label{fig:vctk-stoi}
\end{figure}

\begin{figure}[h]
\centering
\includegraphics[width=\linewidth]{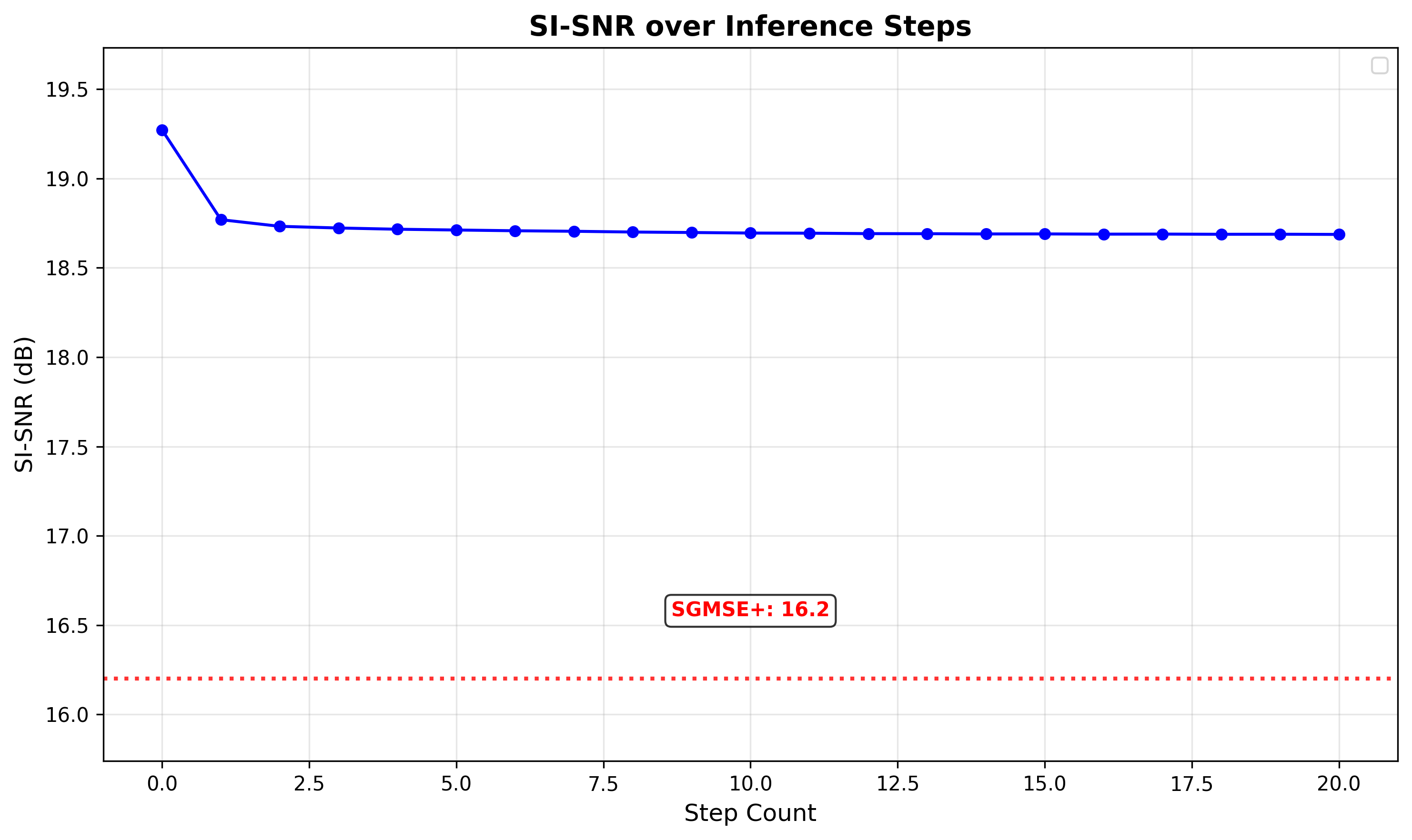}
\caption{SI-SNR across all inference steps for VCTK-DEMAND test set.}
\label{fig:vctk-sisnr}
\end{figure}

\subsection{Results on DNS Challenge V3 test set}
We report the per-step UTMOS, DNSMOS SIG, BAK, OVRL and MOS results for DNS Challenge V3 test set in Figure~\ref{fig:dns-utmos}, Figure~\ref{fig:dns-sig}, Figure~\ref{fig:dns-bak}, Figure~\ref{fig:dns-ovrl} and Figure~\ref{fig:dns-mos}. Similar to VCTK-DEMAND, we notice the results become stable after two inference steps, where the first inference step yielded the most improvements.

\begin{figure}[h]
\centering
\includegraphics[width=\linewidth]{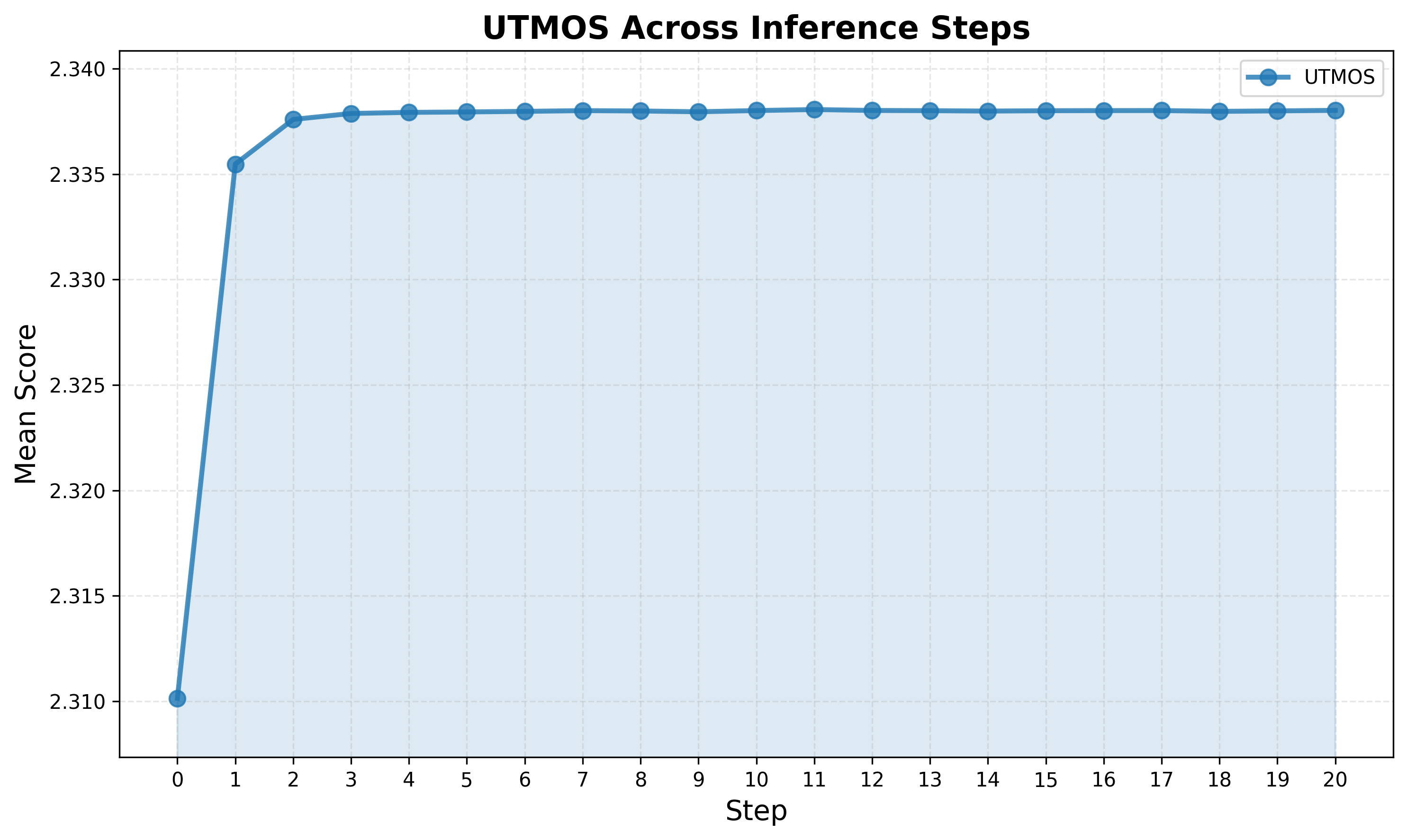}
\caption{UTMOS across all inference steps for DNS Challenge V3 test set.}
\label{fig:dns-utmos}
\end{figure}

\begin{figure}[h]
\centering
\includegraphics[width=\linewidth]{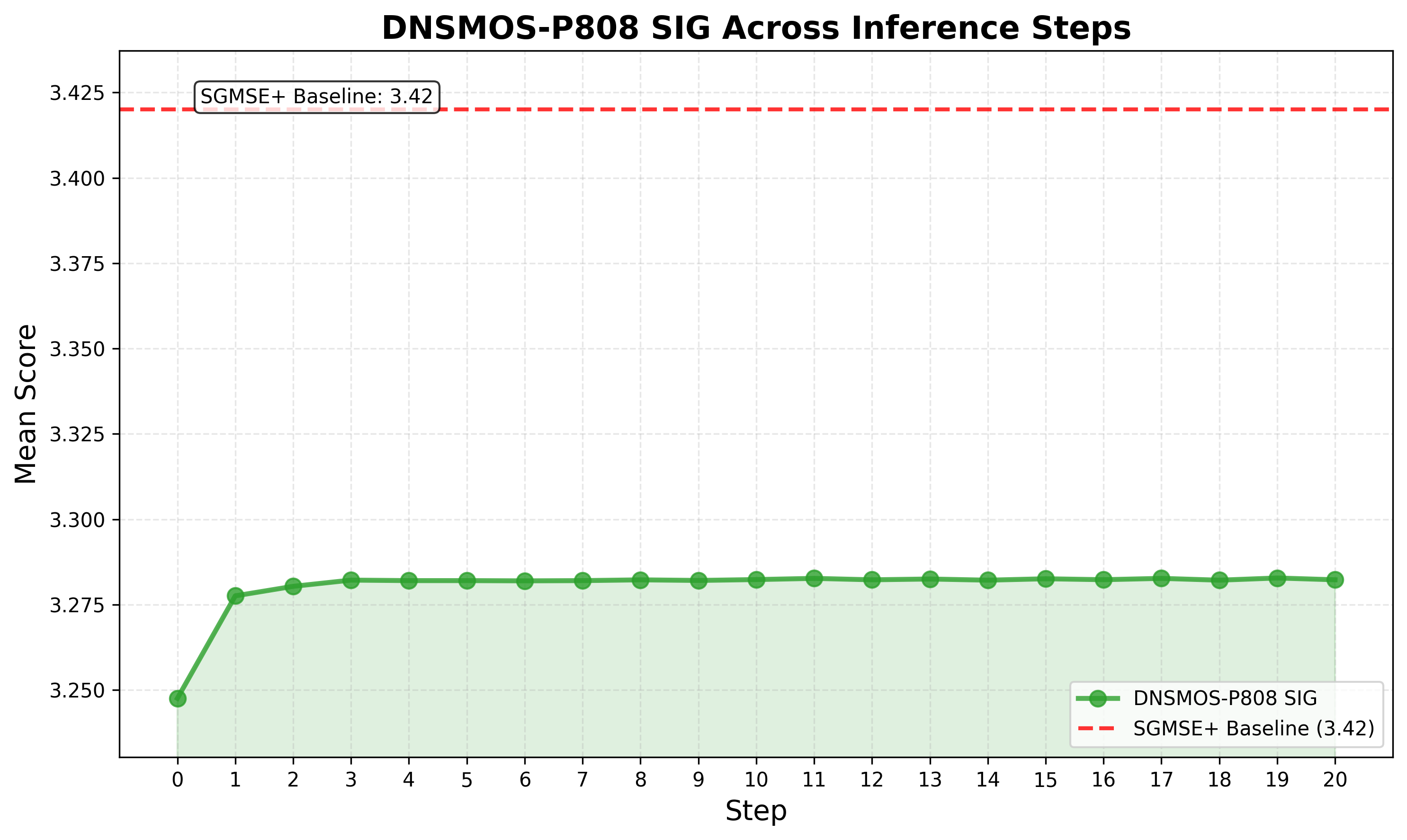}
\caption{DNSMOS SIG across all inference steps for DNS Challenge V3 test set.}
\label{fig:dns-sig}
\end{figure}

\begin{figure}[h]
\centering
\includegraphics[width=\linewidth]{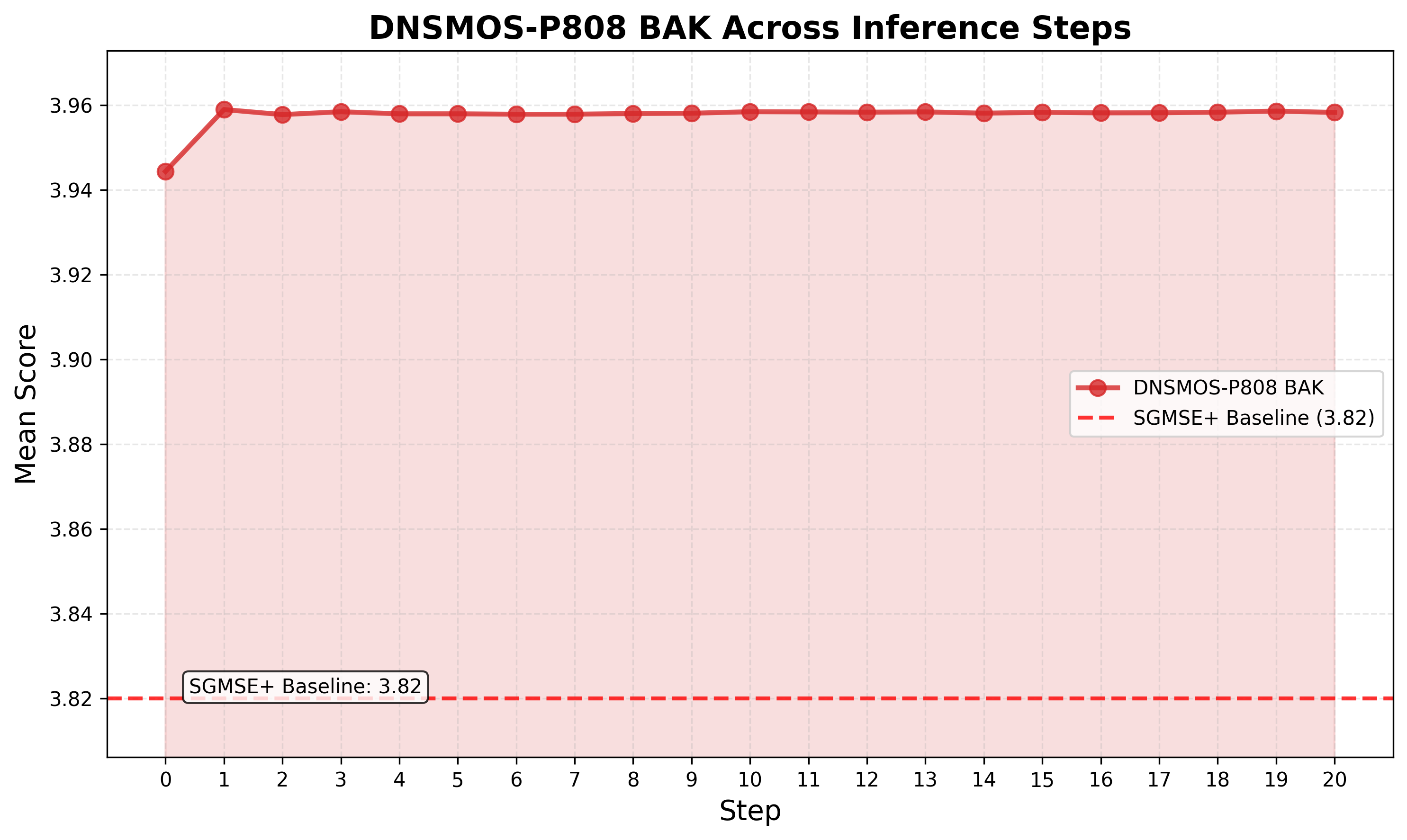}
\caption{DNSMOS BAK across all inference steps for DNS Challenge V3 test set.}
\label{fig:dns-bak}
\end{figure}

\begin{figure}[h]
\centering
\includegraphics[width=\linewidth]{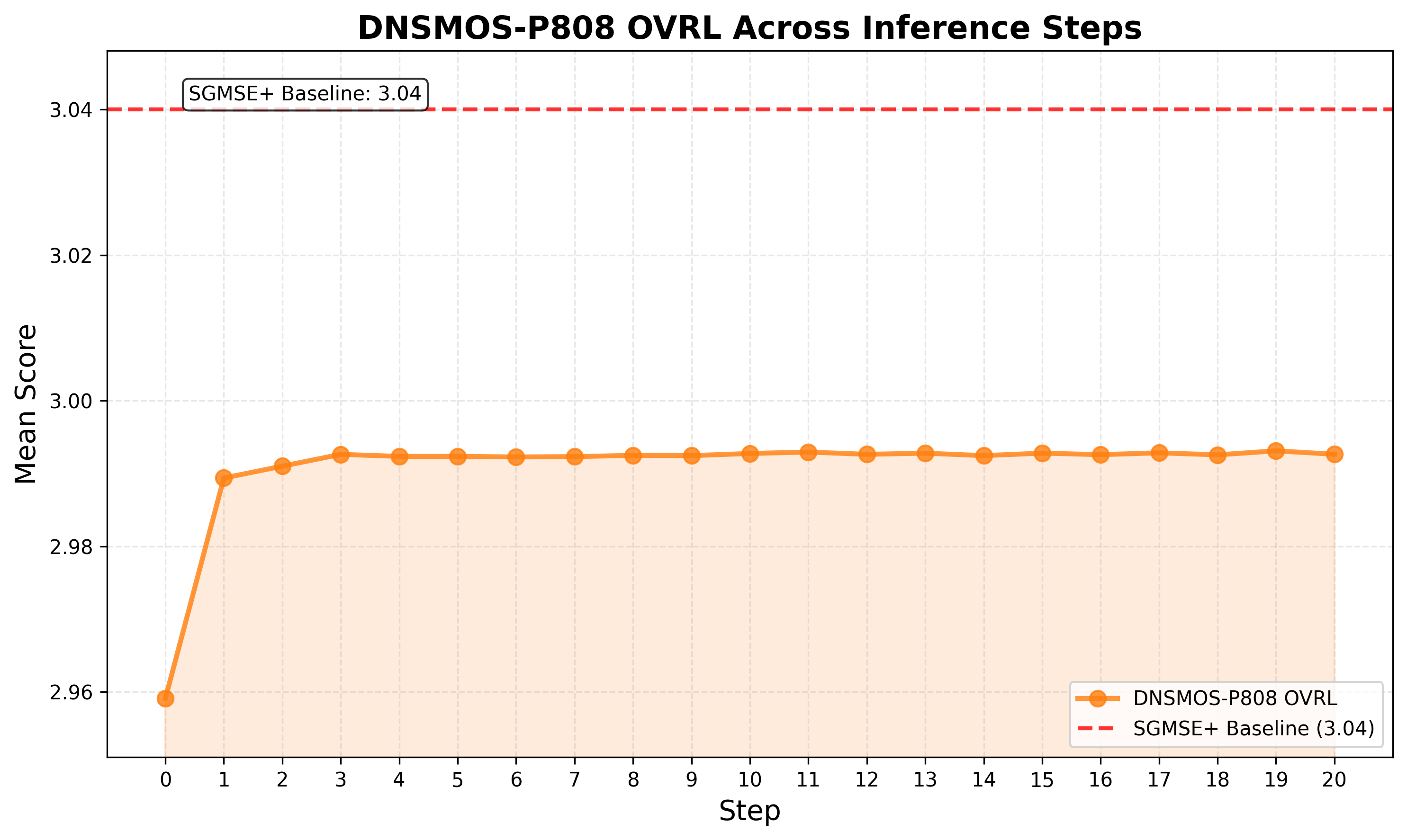}
\caption{DNSMOS OVRL across all inference steps for DNS Challenge V3 test set.}
\label{fig:dns-ovrl}
\end{figure}

\begin{figure}[h]
\centering
\includegraphics[width=\linewidth]{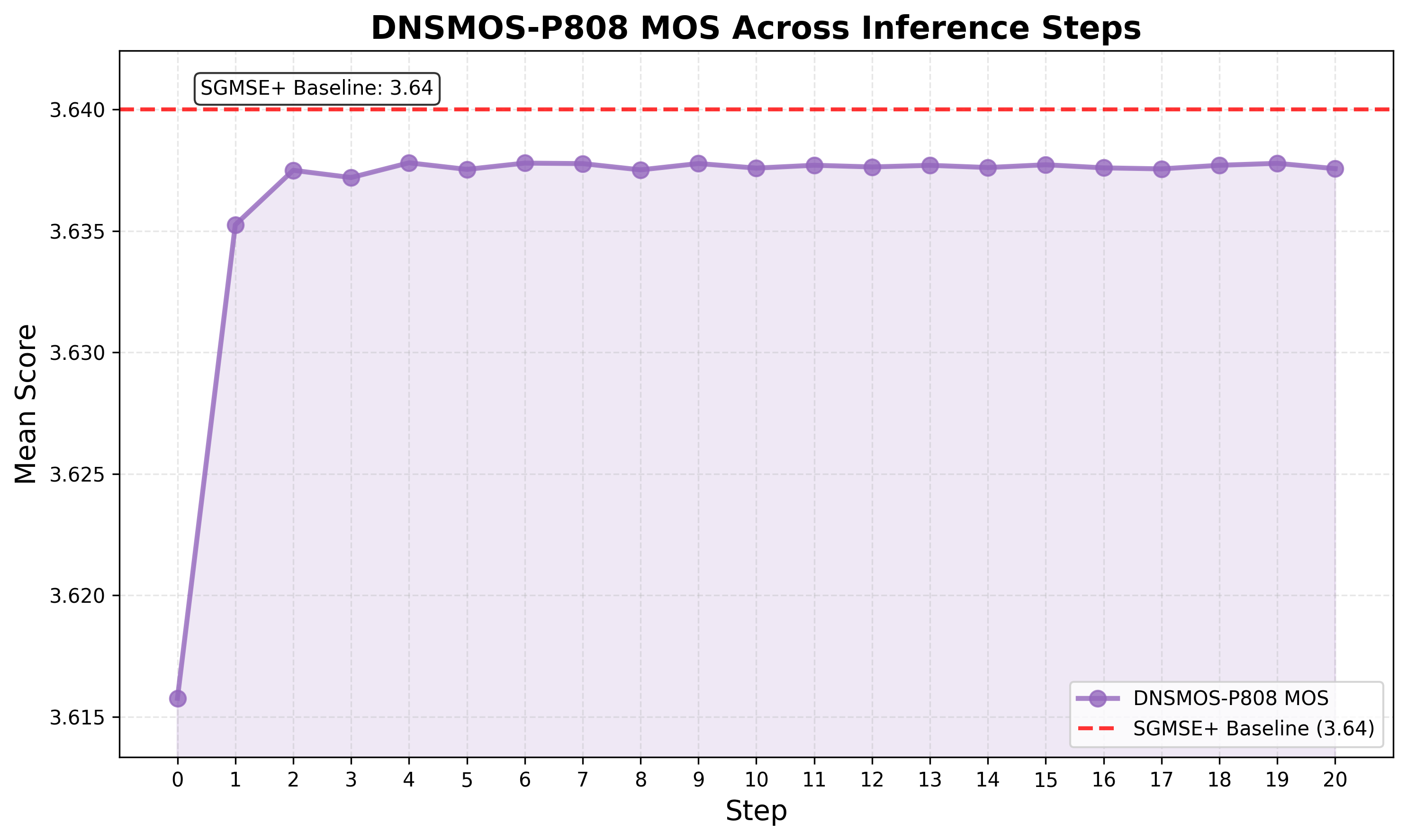}
\caption{DNSMOS MOS across all inference steps for DNS Challenge V3 test set.}
\label{fig:dns-mos}
\end{figure}

\clearpage
\section{Additional Details for Music Source Separation}
Following the approach used for speech enhancement, we present pre-step uSDR and cSDR results for each of the four stems individually in Figure~\ref{fig:mss-vocals}, Figure~\ref{fig:mss-bass}, Figure~\ref{fig:mss-drums} and Figure~\ref{fig:mss-other}. The chunked-SDR metric exhibits more pronounced fluctuations, indicating the inherent instability of taking the median of song-median chunked SDR values. In contrast, the utterance-SDR demonstrates considerably greater stability. Across both metrics, the most substantial improvements occur during the initial inference step, consistent with the behavior observed in many bridge models.

\begin{figure}[h]
\centering
\includegraphics[width=\linewidth]{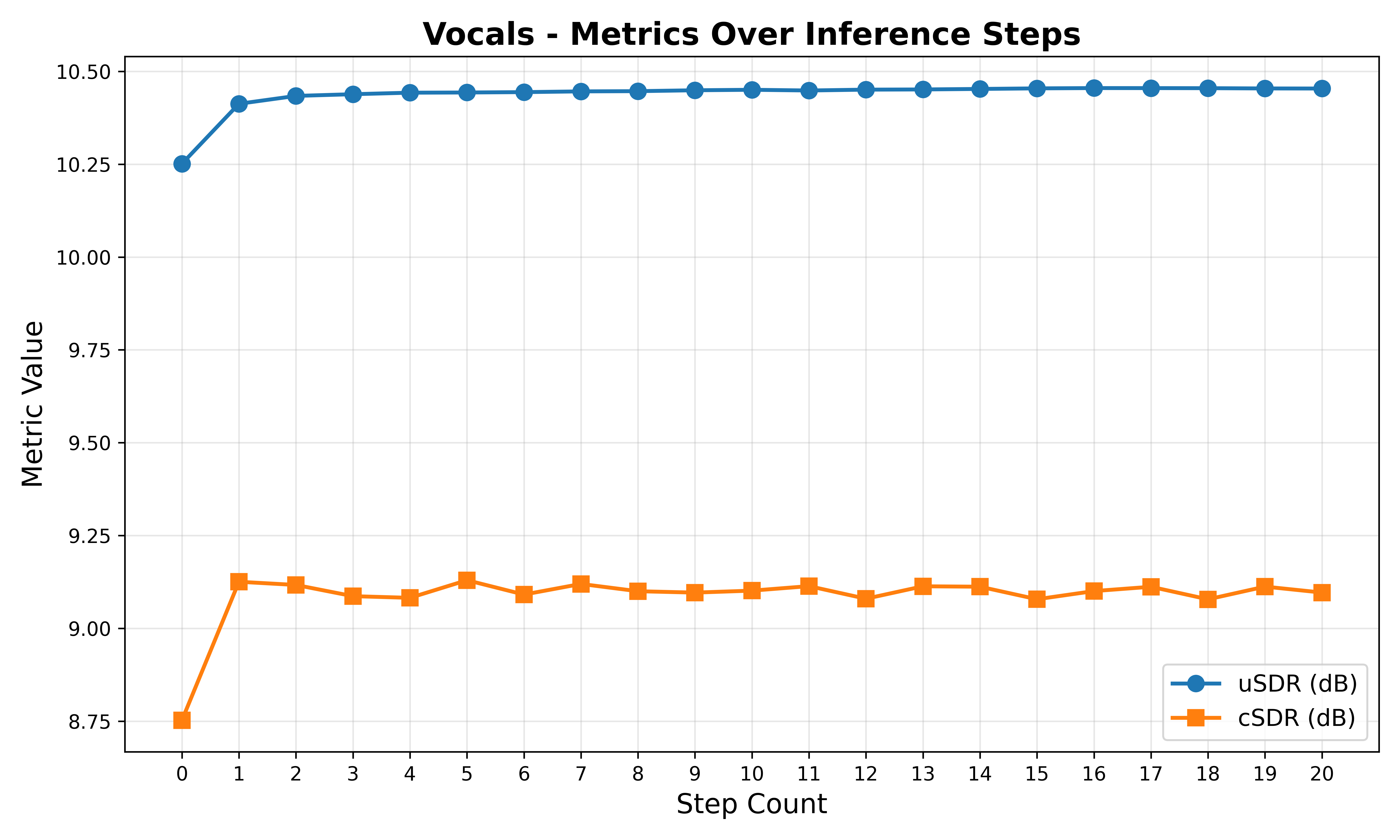}
\caption{uSDR and cSDR performance across all inference steps for ``Vocals'' stem in the music source separation task.}
\label{fig:mss-vocals}
\end{figure}

\begin{figure}[h]
\centering
\includegraphics[width=\linewidth]{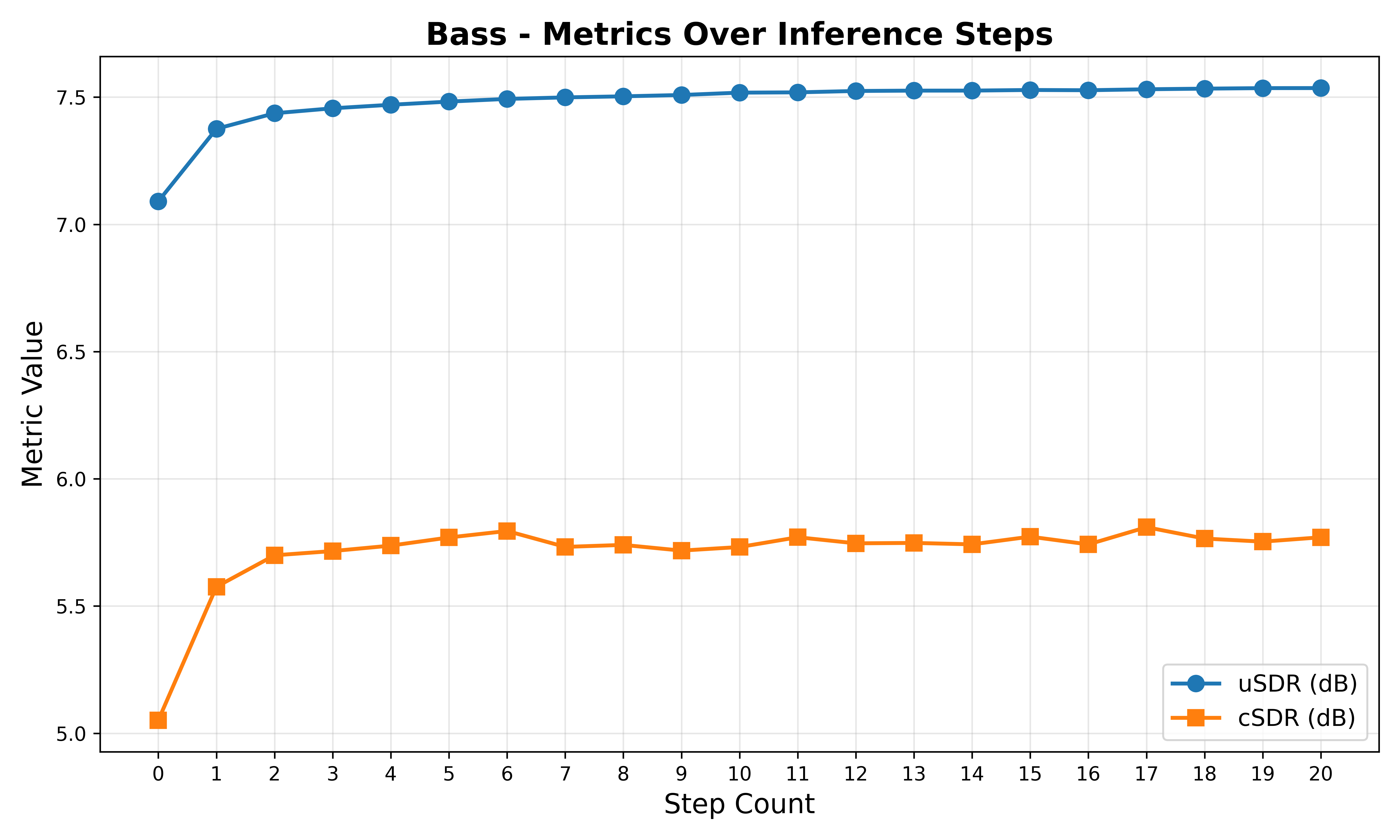}
\caption{uSDR and cSDR performance across all inference steps for ``Bass'' stem in the music source separation task.}
\label{fig:mss-bass}
\end{figure}

\begin{figure}[h]
\centering
\includegraphics[width=\linewidth]{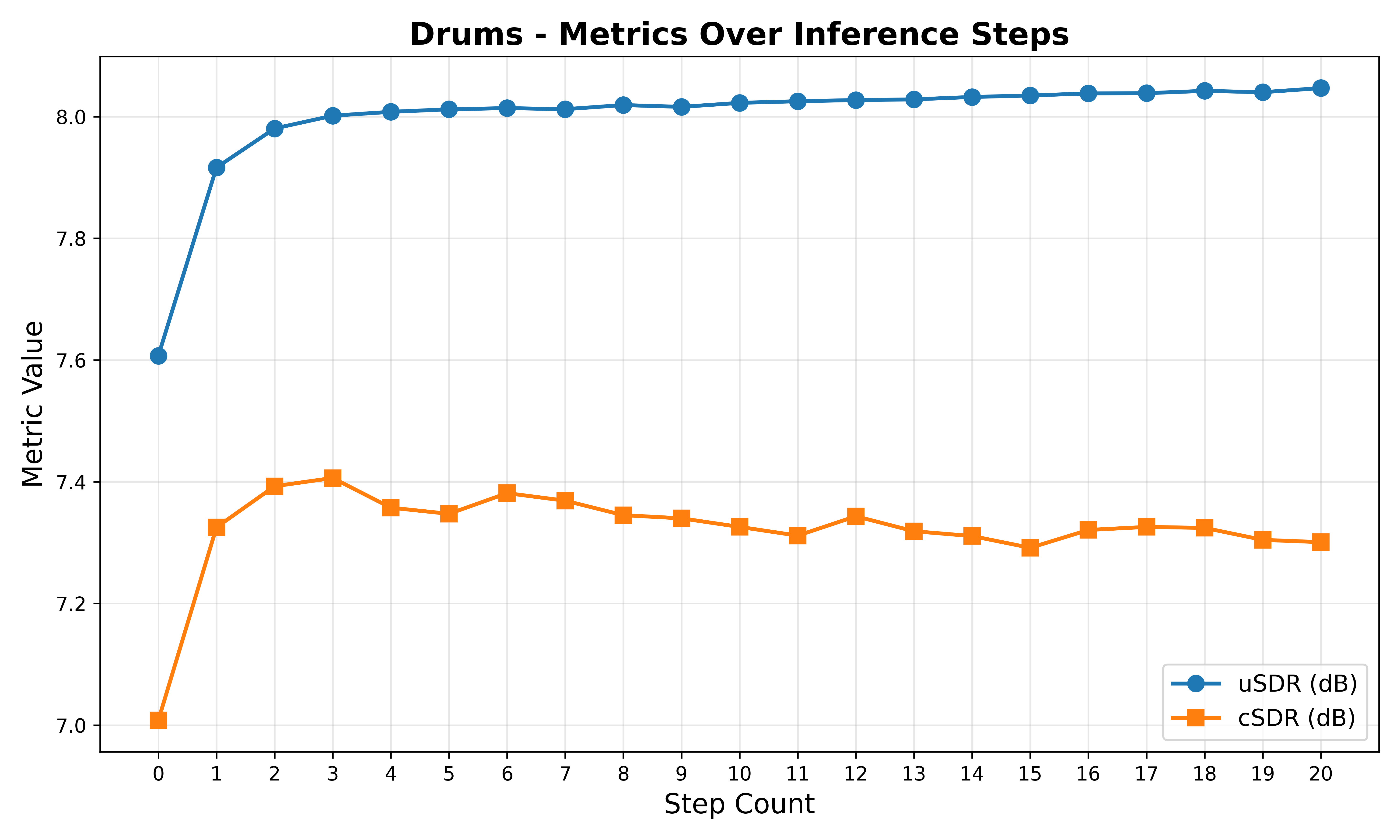}
\caption{uSDR and cSDR performance across all inference steps for ``Drums'' stem in the music source separation task.}
\label{fig:mss-drums}
\end{figure}

\begin{figure}[h]
\centering
\includegraphics[width=\linewidth]{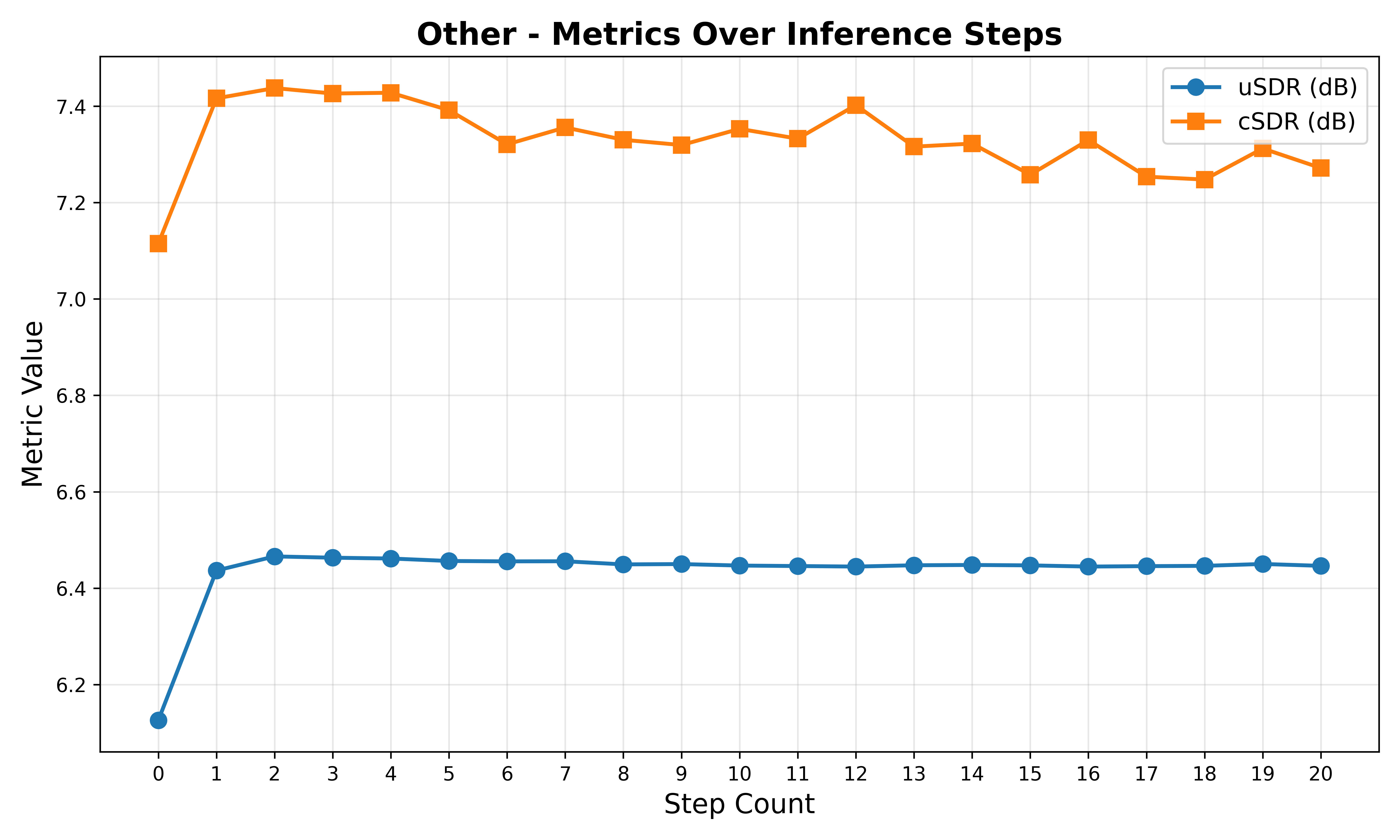}
\caption{uSDR and cSDR performance across all inference steps for ``Other'' stem in the music source separation task.}
\label{fig:mss-other}
\end{figure}

\clearpage
\section{Additional Details for Bridge Model Connection}
\subsection{Justification for Smoothing the Dirac delta}
In Section 5, we claimed that the Dirac delta in the DDBM framework can be approximated by a Gaussian with small variance. We now provide a justification for this approach.

Starting with the deterministic linear bridge:
\begin{equation}
x_t = (1-t)x_0 + tx_T
\end{equation}

To enable differentiability required for score matching, we add infinitesimal Gaussian noise:
\begin{equation}
x_t^{\epsilon} = (1-t)x_0 + tx_T + \epsilon \cdot z, \quad z \sim \mathcal{N}(0, I)
\end{equation}

This yields the conditional distribution:
\begin{equation}
q_{\epsilon}(x_t|x_0, x_T) = \mathcal{N}(x_t; (1-t)x_0 + tx_T, \epsilon^2 I)
\end{equation}

As $\epsilon \to 0$, this distribution converges to the Dirac delta in the weak topology:
\begin{equation}
\lim_{\epsilon \to 0} q_{\epsilon}(x_t|x_0, x_T) = \delta(x_t - ((1-t)x_0 + tx_T))
\end{equation}

However, for any $\epsilon > 0$, the distribution remains differentiable, enabling score computation:
\begin{equation}
\nabla_{x_t} \log q_{\epsilon}(x_t|x_0, x_T) = -\frac{1}{\epsilon^2}(x_t - ((1-t)x_0 + tx_T))
\end{equation}

Digital audio typically uses 16-bit representation, introducing an inherent noise floor at approximately -96 dB. This corresponds to $\epsilon^2 \approx 10^{-9.6}$ relative to full scale, providing a natural lower bound for the variance in our Gaussian approximation. Thus, the Gaussian smoothing with small but non-zero variance accurately models the physical reality of audio signals, where perfect deterministic values are never achieved in practice.

\subsection{Justification for Score Parameterization}

We now derive the score parameterization used in our method from first principles, showing how it naturally emerges from the audio separation objective.

Starting from the score definition with Gaussian smoothing:
\begin{equation}
\nabla_x \log q_{\epsilon}(x|x_0, x_T) = -\frac{1}{\epsilon^2}(x - \mu(x_0, x_T))
\end{equation}
where $\mu(x_0, x_T) = (1-t)x_0 + tx_T$.

For audio separation, we re-parameterize in terms of the clean signal estimate. Let $\hat{p}(x)$ denote the model's estimate of the clean signal from mixture $x$.

At optimal training, the expected value of the estimate equals the true clean signal:
\begin{equation}
\mathbb{E}[\hat{p}(x_t)] = p
\end{equation}

For the linear bridge with $x_0 = n$ (noise) and $x_T = p$ (clean signal):
\begin{equation}
\nabla_x \log q_{\epsilon} = -\frac{1}{\epsilon^2}(x - ((1-t)n + tp))
\end{equation}

For a well-trained model, we have the approximation:
\begin{equation}
\hat{p}(x) - x \approx t(p - x)
\end{equation}

This approximation holds because the model learns to estimate how much of the clean signal is present in the mixture, which is proportional to $t$ in our linear bridge.

Rearranging and substituting:
\begin{align}
s_{\omega}(x, p, t) &= -\frac{1}{\epsilon^2 t}(\hat{p}(x) - x) \\
&= \frac{1}{\epsilon^2 t}(x - \hat{p}(x))
\end{align}

In our notation where $t = \sigma$, this becomes:
\begin{equation}
s_{\omega}(y, p, \sigma) = \frac{1}{\epsilon^2 \sigma}(y - \hat{p}(y))
\end{equation}

This naturally explains the $1/\sigma$ factor in Equation (21) of the main paper, showing that our score parameterization emerges directly from the bridge diffusion framework when applied to audio separation.

\end{document}